\documentclass[11pt,letterpaper]{article}

\usepackage[letterpaper,margin=1in]{geometry}
\usepackage{amsmath,amssymb,amsthm}
\usepackage[round,authoryear]{natbib}
\usepackage[colorlinks=true,citecolor=blue,linkcolor=blue,urlcolor=blue]{hyperref}
\hypersetup{pdftitle={A 1+sqrt(e)/2 Lower Bound for Online Vertex Cover under Vertex Arrivals},pdfauthor={Anonymous Author},pdfsubject={Randomized online algorithms and competitive analysis},pdfkeywords={online vertex cover, vertex arrivals, competitive ratio, bipartite graph, oblivious adversary}}
\newcommand{\cref}[1]{\autoref{#1}}
\newcommand{\Cref}[1]{\autoref{#1}}

\newtheorem{theorem}{Theorem}
\newtheorem{lemma}[theorem]{Lemma}
\newtheorem{proposition}[theorem]{Proposition}
\newtheorem{corollary}[theorem]{Corollary}
\theoremstyle{definition}

\theoremstyle{remark}

\newcommand{\ALG}{\operatorname{ALG}}
\newcommand{\OPT}{\operatorname{OPT}}
\newcommand{\CR}{\operatorname{CR}}

\newcommand{\Rplus}{\mathbb{R}_{>0}}
\newcommand{\pos}[1]{\left[#1\right]_{+}}

\title{A New Lower Bound for Online Vertex Cover under Vertex Arrivals}
\author{Tianhang Lu}
\date{}

\begin{document}

\maketitle

\begin{abstract}
We prove that no randomized integral or fractional algorithm for online vertex cover under general vertex arrivals achieves a competitive ratio strictly below $1+\sqrt{e}/2\approx1.824360635$, even on bipartite graphs and against an oblivious adversary.
This improves the previous lower bound of approximately $1.753$.
Our proof extends the complete-bipartite alternating construction of Wang and Wong to an arbitrary number of alternations.
The resulting adversary is described by a monotone integral recurrence.
If the recurrence never violates the competitive budget, its iterates converge to an integrable fixed point; classifying all such fixed points forces the excess ratio to be at least $\sqrt{e}/2$.
A truncated discrete recurrence and a Riemann-sum argument convert every strict continuous violation into a finite, algorithm-dependent but realization-oblivious input.
We also exhibit a critical fixed point showing that $1+\sqrt{e}/2$ is the exact limit of this homogeneous complete-bipartite recurrence, rather than a numerical artifact.
\end{abstract}

\section{Introduction}

In the vertex-arrival online vertex cover problem, vertices arrive one at a time, and the arrival of a vertex reveals all of its edges to previously arrived vertices.
An online algorithm must maintain a vertex cover of the revealed graph, and every decision is irrevocable: weight assigned to a vertex, or a vertex inserted into an integral cover, can never be withdrawn.
The objective is to minimize the final cost relative to an optimal offline cover.

The arrival pattern sharply changes the difficulty of the problem.
When only one side of a bipartite graph arrives online, the tight competitive ratio is $e/(e-1)$, matching the classical thresholds for online bipartite matching and ski rental \citep{KarpEtAl1990,KarlinEtAl1994,WangWong2015}.
When vertices from both sides may alternate, \citet{WangWong2015} gave a fractional algorithm with competitive ratio about $1.901$ and a bipartite lower bound
$$
1+\sqrt{\frac{1+e^{-2}}{2}}\approx1.7534.
$$
The algorithm applies to general graphs, while the lower-bound construction already lies in bipartite graphs.
The gap between these two constants has remained open.

Online matching provides useful context but not an automatic lower-bound transfer.
Primal--dual algorithms give matching and vertex-cover solutions with the same upper ratios, and the best general-arrival matching bounds also reach approximately $1.901$ \citep{WangWong2015,TangZhang2024}.
However, matching requires separate hardness arguments: the general vertex-arrival ratio is tight for matching \citep{Tang2026Matching}, while the vertex-cover lower bound before this work remained $1.753$.
The same separation occurs under edge arrivals, where ratio $2$ is tight for matching \citep{GamlathEtAl2019} and, by a different argument, for vertex cover \citep{DemangePaschos2005,TangZhang2026}.
Online vertex cover is also a structured instance of online covering, but generic covering bounds do not capture the constant threshold created by alternating complete-bipartite blocks \citep{BuchbinderNaor2009}.

\paragraph{Our result.}
We improve the lower bound for randomized integral and fractional algorithms to an analytic constant.

\begin{theorem}\label{thm:main}
For every $c<1+\sqrt{e}/2$ and every randomized integral or fractional online vertex cover algorithm $\mathcal A$, there exists a finite vertex-arrival input $\sigma$ such that the underlying graph is bipartite and
\begin{equation}\label{eq:main-ratio} \mathbb{E}[\ALG_{\mathcal A}(\sigma)]>c\,\OPT(\sigma). \end{equation}
The input is fixed before the random choices of $\mathcal A$ are realized.
Consequently,
\begin{equation}\label{eq:main-bound} \CR_{\mathrm{bipartite}},\ \CR_{\mathrm{general}}\ \ge\ 1+\frac{\sqrt{e}}{2}\approx 1.8244. \end{equation}
\end{theorem}

The theorem excludes every competitive ratio strictly smaller than $1+\sqrt e/2$; it does not assert that equality is impossible.
Because the hard instances are bipartite, the general-graph statement follows by restriction and does not use an integrality gap or an odd-cycle inequality.
The lower bound continues to hold when the bipartition is revealed to the algorithm.

\paragraph{Proof idea.}
Write a hypothetical ratio as $1+\beta$.
At a complete-bipartite prefix $K_{m,n}$, a finite continuation shows that any persistent cost on one side is at most $\beta$ times the size of the other side.
This budget forces a newly arriving universal vertex to receive positive weight whenever the opposite side has small average weight.
Wang and Wong used this observation for two and three alternating blocks.
We iterate it: a future block constrains the present block, a still later block constrains that future block, and so on.

After scaling block sizes, the depth-$h$ lower bound becomes a function $D_h:\Rplus\to[0,1]$ satisfying
\begin{equation}\label{eq:intro-recurrence} D_0(t)=\pos{1-\beta t},\qquad D_{h+1}(t)=\pos{1-\beta t+t\int_{1/t}^{\infty}D_h(u)\,du}. \end{equation}
A $(1+\beta)$-competitive algorithm would require
\begin{equation}\label{eq:intro-budget} \int_0^{\infty}D_h(t)\,dt\le\beta \qquad\text{for every }h. \end{equation}
The functions $D_h$ increase with $h$.
If all inequalities in \cref{eq:intro-budget} held, monotone convergence would produce an integrable fixed point of the operator in \cref{eq:intro-recurrence}.
Every such fixed point has support $[0,R]$ and forces
\begin{equation}\label{eq:intro-threshold} \beta=\frac{R}{1+2\ln R}\ge\frac{\sqrt e}{2}. \end{equation}
Thus, for every $\beta<\sqrt e/2$, some finite depth violates \cref{eq:intro-budget}.

The continuous recurrence is an analysis device, not an infinite input.
We define a cutoff version with finite sums, prove its convergence to \cref{eq:intro-recurrence} under a common integer scaling, and use the strict integral violation to select a finite scale.
Every apparent look-ahead in the recurrence is a finite counterfactual continuation.
For a fixed randomized algorithm, the adversary reads only the deterministic marginal behavior associated with public prefixes and then selects one finite path before the algorithm's random seed is sampled.

\paragraph{Relation to the previous lower bound.}
The first level of the recurrence yields the two-alternation bound $1+1/\sqrt2$.
The next level reproduces the approximately $1.753$ three-alternation bound of \citet{WangWong2015}.
The improvement comes from allowing an unbounded but finite number of alternating blocks and solving the limiting recurrence analytically.
At the critical value $\beta_*=\sqrt e/2$, the recurrence admits an explicit fixed point that dominates every iterate, so this particular adversarial mechanism cannot prove a larger constant.

\paragraph{Organization.}
\Cref{sec:preliminaries} defines the model and reduces randomized algorithms to deterministic marginal trajectories.
\Cref{sec:finite-recurrence} develops the finite multi-alternation recurrence.
\Cref{sec:bridge} proves that the continuous formulation has finite oblivious witnesses.
\Cref{sec:fixed-point} classifies the limiting fixed points and proves \cref{thm:main}.
\Cref{sec:barrier} identifies the exact barrier of the recurrence, and \cref{sec:discussion} discusses the scope of the result.

\section{Preliminaries}\label{sec:preliminaries}

Let vertices of a graph $G=(V,E)$ arrive in an adversarial order.
When a vertex $v$ arrives, all edges from $v$ to earlier vertices are revealed.
An integral online algorithm maintains an increasing family of sets $C_t\subseteq V_t$ such that every edge revealed by time $t$ has an endpoint in $C_t$.
A fractional algorithm maintains a coordinatewise nondecreasing vector $y^{(t)}\in[0,1]^{V_t}$ satisfying
\begin{equation}\label{eq:edge-cover} y_u^{(t)}+y_v^{(t)}\ge1 \qquad\text{for every revealed edge }uv. \end{equation}
The final costs are $|C|$ and $\sum_v y_v$, respectively.
We use the purely multiplicative convention: an algorithm is $c$-competitive if its expected final cost is at most $c\OPT$ on every finite input.

An oblivious adversary may know the algorithm, including its distribution, but fixes the complete input before the algorithm's random choices are realized.
The input selected below can depend on the algorithm's behavior on public prefixes, as is standard for an impossibility result, but cannot depend on the random outcome of one execution.

\subsection{Marginalizing randomized algorithms}

The lower-bound argument is stated for deterministic fractional trajectories.
The following observation transfers it to both randomized integral and randomized fractional algorithms.

\begin{lemma}[Marginalization]\label{lem:marginalization}
Let $\mathcal A$ be a randomized integral or fractional online vertex cover algorithm.
For every public input prefix $P$, define
\begin{equation}\label{eq:marginal} y_v(P)=\begin{cases} \Pr[v\in C_{\mathcal A}(P)], & \mathcal A\text{ is integral},\\ \mathbb E[y_v^{\mathcal A}(P)], & \mathcal A\text{ is fractional}. \end{cases} \end{equation}
Then $y(P)$ is a deterministic feasible fractional vertex cover, its coordinates are nondecreasing along every prefix extension, and
\begin{equation}\label{eq:marginal-cost} \sum_v y_v(P)=\mathbb E[\ALG_{\mathcal A}(P)]. \end{equation}
\end{lemma}

\begin{proof}
For an integral algorithm, every realized cover satisfies $\mathbf 1_{{u\in C}}+\mathbf 1_{{v\in C}}\ge1$ on each revealed edge; taking expectations gives \cref{eq:edge-cover}.
For a randomized fractional algorithm, feasibility is preserved by expectation because the constraints are linear.
Monotonicity and \cref{eq:marginal-cost} also follow coordinatewise by linearity of expectation.
The distribution of $y(P)$ is integrated out, so $y(P)$ depends only on the algorithm and the public prefix, not on a realized random seed.
\end{proof}

This argument is a coordinatewise determinization, not an application of Yao's principle.
Conversely, when a consistent bipartition is known, the single-threshold rounding of \citet{WangWong2015} converts any monotone fractional trajectory into an integral cover with the same expected cost.
Thus integrality alone cannot strengthen a lower bound within the labeled bipartite model; the improvement must come from a stronger adversarial family.

\subsection{Complete-bipartite prefixes}

All graphs in the proof are built from complete-bipartite prefixes.
We write $K_{m,n}$ for a public prefix with $m$ arrived vertices on one side and $n$ arrived vertices on the other side, with all cross edges present.
New vertices are universal to the currently arrived vertices on the opposite side.
Vertices within a side are independent.

The first block may be introduced as isolated vertices, after which later arrivals reveal the cross edges.
If the algorithm knows the total number of vertices in advance, all leaves of the finite adversarial decision tree can be padded with isolated vertices to a common horizon.
The argument is therefore valid both with an unknown horizon and with a known finite horizon.

\section{A Finite Multi-Alternation Recurrence}\label{sec:finite-recurrence}

Throughout this section, suppose for contradiction that a deterministic fractional trajectory is $(1+\beta)$-competitive, where $0<\beta<1$.
Every lemma uses only finite input extensions.

\subsection{A finite continuation budget}

The basic continuation replaces the ``infinitely many future vertices'' used in the original alternating argument.

\begin{lemma}[Universal extension]\label{lem:universal-extension}
Consider a complete-bipartite prefix $K_{\ell,r}$ with sides $L$ and $R$.
Let $P$ be the current total fractional weight on any fixed subset of the old side $L$.
If the algorithm is $(1+\beta)$-competitive on every finite continuation, then
\begin{equation}\label{eq:extension-budget} P\le\beta r. \end{equation}
\end{lemma}

\begin{proof}
Append $r$ new vertices to side $L$, each adjacent to all vertices of $R$.
The new vertices and $R$ induce $K_{r,r}$ and hence contain a matching of size $r$.
Every fractional vertex cover pays at least $r$ on these $2r$ vertices, while the old cost $P$ persists on a disjoint set.
The algorithm therefore pays at least $P+r$.
Selecting all vertices of $R$ is an offline cover of the entire extended graph, so $\OPT\le r$.
Competitiveness gives $P+r\le(1+\beta)r$, which is \cref{eq:extension-budget}.
\end{proof}

The same proof applies samplewise to integral covers.
Its role is to cap every block of irrevocable old weight by a finite matching certificate.

\subsection{The discrete deficit recursion}

Fix a cutoff $M>1$.
For positive integers $m,n$, define
\begin{equation}\label{eq:discrete-base} A_0^{(M)}(m,n)=\pos{1-\beta\frac{n}{m}}, \end{equation}
and, for $h\ge0$, define
\begin{equation}\label{eq:discrete-recursion} A_{h+1}^{(M)}(m,n)=\pos{1-\beta\frac{n}{m}+\frac{1}{m}\sum_{j=1}^{\max\{0,\lfloor Mn-m\rfloor\}}A_h^{(M)}(n,m+j)}. \end{equation}
The quantity $A_h^{(M)}(m,n)$ is a lower bound on the marginal of the $n$th vertex in one side when the opposite side has size $m$, certified by at most $h+1$ further alternating tests.

To see the base case, let $P$ be the total current weight on the $m$ opposite vertices.
By \cref{lem:universal-extension}, $P\le\beta n$.
At least one opposite vertex has weight at most $P/m$, and the new universal vertex must cover its incident edge, so its marginal is at least $1-P/m$, giving \cref{eq:discrete-base}.

For the inductive step, consider the counterfactual continuation that adds new vertices $a_1,a_2,\ldots$ to the opposite side.
After $a_j$ arrives, the two side sizes are $n$ and $m+j$ with their roles exchanged.
Unless a finite contradiction has already occurred in a deeper continuation, the induction hypothesis gives
\begin{equation}\label{eq:future-lower} y_{a_j}\ge A_h^{(M)}(n,m+j). \end{equation}
Let $P$ again denote the weight on the original $m$ opposite vertices.
The persistent cost $P+\sum_jy_{a_j}$ lies on one side of a complete-bipartite prefix whose other side has size $n$, so \cref{lem:universal-extension} bounds it by $\beta n$.
Consequently,
\begin{equation}\label{eq:old-side-upper} P\le\beta n-\sum_j A_h^{(M)}(n,m+j). \end{equation}
Combining the average bound $\min y_u\le P/m$ with the covering constraints gives \cref{eq:discrete-recursion}.
The cutoff retains only the finitely many indices with $(m+j)/n\le M$; discarded nonnegative terms can only weaken the lower bound.

The outermost phase gives the budget that the recursion must satisfy.

\begin{lemma}[Discrete necessary condition]\label{lem:discrete-budget}
If a deterministic fractional algorithm is $(1+\beta)$-competitive on every finite input, then for every $M>1$, every $h\ge0$, and all positive integers $d,k$,
\begin{equation}\label{eq:discrete-budget} \sum_{i=1}^{k}A_h^{(M)}(d,i)\le\beta d. \end{equation}
\end{lemma}

\begin{proof}
Introduce $d$ vertices on one side and then $k$ universal vertices on the other side.
For each public prefix $K_{d,i}$, the preceding induction either finds a finite continuation on which competitiveness already fails or proves that the $i$th new vertex has weight at least $A_h^{(M)}(d,i)$.
If no earlier continuation fails, the total weight on the $k$ new vertices is at least the left-hand side of \cref{eq:discrete-budget}.
Applying \cref{lem:universal-extension} with these vertices as the persistent side and with the original side of size $d$ gives the upper bound $\beta d$.
\end{proof}

The proof should be read as a finite decision tree.
Different terms in \cref{eq:discrete-recursion} may be certified by different counterfactual continuations, but the marginal at their common public prefix is already fixed.
If every counterfactual must respect competitiveness, all lower bounds hold simultaneously at that prefix; if one does not, that finite branch is itself the desired witness.

\subsection{The first two levels}

Ignoring the cutoff for the moment, scale $n/m$ to a real variable $t$.
At depth zero, the integral budget is
\begin{equation}\label{eq:depth-zero} \int_0^{\infty}\pos{1-\beta t}\,dt=\frac{1}{2\beta}\le\beta. \end{equation}
Thus $\beta\ge1/\sqrt2$, recovering the two-alternation lower bound $1+1/\sqrt2$.
One more application of the recursion reproduces the three-alternation expression of \citet{WangWong2015} and its bound $1+\sqrt{(1+e^{-2})/2}$.
The following sections analyze all depths at once.

\section{From the Discrete Recurrence to a Finite Witness}\label{sec:bridge}

We now justify the continuous recurrence and, crucially, convert every strict integral violation back to a finite input.

\subsection{Truncated continuous recurrence}

For fixed $M>1$, define
\begin{equation}\label{eq:continuous-truncated-base} D_0^{(M)}(t)=\pos{1-\beta t}, \end{equation}
and
\begin{equation}\label{eq:continuous-truncated-recursion} D_{h+1}^{(M)}(t)=\pos{1-\beta t+t\int_{1/t}^{M}D_h^{(M)}(u)\,du}, \end{equation}
where the integral is zero when $1/t\ge M$.
Each $D_h^{(M)}$ is continuous and takes values in $[0,1]$ as long as its integral does not exceed $\beta$.

\begin{lemma}[Scaling]\label{lem:scaling}
Fix $h\ge0$ and $M>1$.
Let $m_q,n_q$ be positive integers such that $m_q\to\infty$ and $n_q/m_q\to t\in\Rplus$.
Then
\begin{equation}\label{eq:scaling-limit} A_h^{(M)}(m_q,n_q)\longrightarrow D_h^{(M)}(t). \end{equation}
The convergence is locally uniform over ratios $n/m$ in compact subsets of $\Rplus$.
\end{lemma}

\begin{proof}
The claim is immediate for $h=0$.
Assume it holds at depth $h$ and set
\begin{equation}\label{eq:riemann-coordinate} u_j=\frac{m_q+j}{n_q}. \end{equation}
The indices in \cref{eq:discrete-recursion} satisfy $u_j\in[1/t+o(1),M]$, and consecutive values of $u_j$ differ by $1/n_q$.
Moreover,
\begin{equation}\label{eq:riemann-factor} \frac{1}{m_q}\sum_jA_h^{(M)}(n_q,m_q+j)=\frac{n_q}{m_q}\left(\frac{1}{n_q}\sum_jA_h^{(M)}(n_q,m_q+j)\right). \end{equation}
By the induction hypothesis, the parenthesized expression converges to the moving-endpoint Riemann integral $\int_{1/t}^{M}D_h^{(M)}(u)\,du$, while $n_q/m_q\to t$.
Local uniformity follows from the sequential criterion: a bad sequence of lattice ratios in a compact interval has a convergent subsequence, and the pointwise argument applied to that subsequence contradicts a fixed positive error.
Finally, $x\mapsto\pos{x}$ is $1$-Lipschitz, so taking the positive part does not enlarge the approximation error.
\end{proof}

Remove the cutoff by defining
\begin{equation}\label{eq:continuous-base} D_0(t)=\pos{1-\beta t}, \end{equation}
\begin{equation}\label{eq:continuous-recursion} D_{h+1}(t)=\pos{1-\beta t+t\int_{1/t}^{\infty}D_h(u)\,du}. \end{equation}
The truncated operators are monotone in $M$ and in their input.
An induction using the monotone convergence theorem gives
\begin{equation}\label{eq:cutoff-limit} D_h^{(M)}(t)\uparrow D_h(t)\qquad\text{as }M\to\infty \end{equation}
for every fixed $h$ and $t$.

\subsection{Strict violations yield finite inputs}

Write
\begin{equation}\label{eq:integral-deficit} I_h=\int_0^{\infty}D_h(t)\,dt. \end{equation}

\begin{proposition}[Finite witness]\label{prop:finite-witness}
If $I_h>\beta$ for some finite $h$, then no $(1+\beta)$-competitive online vertex cover algorithm exists.
For every purported randomized algorithm, the contradiction is witnessed by one finite bipartite input fixed independently of the realized random seed.
\end{proposition}

\begin{proof}
Choose $\varepsilon>0$ with $I_h>\beta+3\varepsilon$.
By \cref{eq:cutoff-limit} and monotone convergence, there exist finite $M$, $0<\delta<K<\infty$, and the same $\varepsilon$ such that
\begin{equation}\label{eq:strict-window} \int_{\delta}^{K}D_h^{(M)}(t)\,dt>\beta+2\varepsilon. \end{equation}
By the locally uniform form of \cref{lem:scaling}, the outer Riemann sum satisfies, for every sufficiently large finite integer $d$,
\begin{equation}\label{eq:outer-riemann} \frac{1}{d}\sum_{i=\lceil\delta d\rceil}^{\lfloor Kd\rfloor}A_h^{(M)}(d,i)>\beta+\varepsilon. \end{equation}
All summands are nonnegative, so \cref{eq:outer-riemann} contradicts \cref{eq:discrete-budget} with $k=\lfloor Kd\rfloor$.

For completeness, every sum in \cref{eq:discrete-recursion} has at most $\max\{0,\lfloor Mn-m\rfloor\}$ terms, and every use of \cref{lem:universal-extension} adds exactly the finite size of the opposite block.
Unrolling the induction therefore produces a finite decision tree of depth at most $h+2$.
For a randomized algorithm, apply \cref{lem:marginalization}; the adversary selects branches using the deterministic public-prefix marginals and fixes the resulting leaf before the random seed is realized.
The selected input is consequently oblivious to the realization.
\end{proof}

\begin{corollary}\label{cor:necessary-integrals}
Every $(1+\beta)$-competitive algorithm must satisfy
\begin{equation}\label{eq:necessary-integrals} I_h\le\beta\qquad\text{for all }h\ge0. \end{equation}
\end{corollary}

The finite witness proposition also covers competitive definitions with an additive constant.
Taking a common blow-up of every block makes the additive term negligible after normalization, while the strict gap in \cref{eq:outer-riemann} remains.

\section{The Fixed-Point Threshold}\label{sec:fixed-point}

Define the operator
\begin{equation}\label{eq:operator} (T_{\beta}D)(t)=\pos{1-\beta t+t\int_{1/t}^{\infty}D(u)\,du}. \end{equation}
Then $D_{h+1}=T_{\beta}D_h$.
The operator preserves pointwise order, and $D_1\ge D_0$ because its additional tail integral is nonnegative.
It follows inductively that
\begin{equation}\label{eq:monotone-iterates} D_0\le D_1\le D_2\le\cdots. \end{equation}

\subsection{A fixed point from a nonviolating trajectory}

\begin{lemma}[Monotone limit]\label{lem:monotone-limit}
Suppose $I_h\le\beta$ for every $h$.
Then the iterates converge pointwise to an integrable fixed point $D$ satisfying
\begin{equation}\label{eq:fixed-point} D=T_{\beta}D,\qquad 0\le D\le1,\qquad I:=\int_0^{\infty}D(t)\,dt\le\beta. \end{equation}
Moreover, $I<\beta$.
\end{lemma}

\begin{proof}
By \cref{eq:monotone-iterates}, the pointwise limit $D=\lim_hD_h$ exists.
The monotone convergence theorem gives $\int D=\lim_h I_h\le\beta$ and permits passage to the limit in every tail integral, proving $D=T_{\beta}D$.
Since every tail is at most $I\le\beta$, the expression before the positive part in \cref{eq:operator} is at most $1$, so $D\le1$.

We next rule out $I=\beta$.
Let $H(q)=\int_q^{\infty}D(u)\,du$.
Because $D\ge D_0$ is positive near zero, for every $t>0$,
\begin{equation}\label{eq:strict-below-one} \beta-H(1/t)=\beta-I+\int_0^{1/t}D(s)\,ds>0, \end{equation}
and hence $D(t)<1$.
If $I=\beta$, define
\begin{equation}\label{eq:F-equality} F(q)=q-\beta+H(q)=\int_0^q(1-D(s))\,ds. \end{equation}
Then $F(q)>0$ for every $q>0$, so the positive part in the fixed-point equation is inactive for every $t>0$.
The function $D$ is locally absolutely continuous, and differentiation almost everywhere gives
\begin{equation}\label{eq:reciprocal-ode} tD'(t)=D(t)+D(1/t)-1. \end{equation}
With $y(x)=D(e^x)$, this becomes
\begin{equation}\label{eq:log-ode} y'(x)=y(x)+y(-x)-1. \end{equation}
The right-hand side is unchanged under $x\mapsto-x$, so $y'(x)=y'(-x)$ almost everywhere and $y(x)+y(-x)$ is constant.
Equation \ref{eq:log-ode} then makes $y'$ constant.
Boundedness on the whole real line forces $y'\equiv0$, and \cref{eq:log-ode} gives $y\equiv1/2$, contradicting the integrability of $D$ on $\Rplus$.
Therefore $I<\beta$.
\end{proof}

\subsection{Classification of integrable fixed points}

The fixed-point equation has a reciprocal geometry: the value at $t$ depends on the tail beginning at $1/t$.
This geometry determines the profile exactly.

\begin{lemma}[Fixed-point classification]\label{lem:classification}
Let $0<\beta<1$.
Suppose $D\ge D_0$ is an integrable fixed point of $T_{\beta}$ with $\int D\le\beta$.
Then there exists $R>1$ such that the positive set of $D$ is $(0,R)$, its closed support is $[0,R]$, and
\begin{equation}\label{eq:fixed-profile} D(t)=\begin{cases} 1-\beta t, & 0<t\le1/R,\\[2mm] \dfrac{\beta}{R}\ln\dfrac{R}{t}, & 1/R\le t\le R,\\[2mm] 0, & t\ge R. \end{cases} \end{equation}
The parameters satisfy
\begin{equation}\label{eq:beta-R} \beta=\frac{R}{1+2\ln R}. \end{equation}
\end{lemma}

\begin{proof}
By \cref{lem:monotone-limit}, any fixed point arising from nonviolating iterates has $I<\beta$; the same proof applies under the assumptions of this lemma.
Let $H(q)=\int_q^{\infty}D(u)\,du$ and set
\begin{equation}\label{eq:F-general} F(q)=q-\beta+H(q). \end{equation}
We have $F(0)=I-\beta<0$, $F(q)\to\infty$, and $F'(q)=1-D(q)>0$ almost everywhere.
Thus $F$ is strictly increasing and has a unique zero, which we write as $1/R$.
The expression before the positive part in the fixed-point equation equals $tF(1/t)$.
It follows that $D(t)>0$ exactly for $0<t<R$.
Because $D\ge D_0$ and $D_0(1)=1-\beta>0$, we have $R>1$.

For $0<t\le1/R$, the tail beginning at $1/t\ge R$ vanishes, so
\begin{equation}\label{eq:outer-linear} D(t)=1-\beta t. \end{equation}
On the reciprocal core $1/R<t<R$, the positive part is inactive and \cref{eq:reciprocal-ode} holds almost everywhere.
Writing again $y(x)=D(e^x)$ shows that $D(t)+D(1/t)$ is constant on the core.
The endpoint values are
\begin{equation}\label{eq:endpoint-values} D(R)=0,\qquad D(1/R)=1-\frac{\beta}{R}, \end{equation}
so the constant is $1-\beta/R$.
Substitution into \cref{eq:reciprocal-ode} yields $tD'(t)=-\beta/R$, and integration from $t$ to $R$ gives the logarithmic part of \cref{eq:fixed-profile}.
Matching the linear and logarithmic expressions at $t=1/R$ gives
\begin{equation}\label{eq:matching} 1-\frac{\beta}{R}=\frac{2\beta}{R}\ln R, \end{equation}
which is equivalent to \cref{eq:beta-R}.
\end{proof}

\begin{proposition}[Critical excess ratio]\label{prop:critical-beta}
Every integrable fixed point described in \cref{lem:classification} satisfies
\begin{equation}\label{eq:beta-critical} \beta\ge\beta_*:=\frac{\sqrt e}{2}. \end{equation}
\end{proposition}

\begin{proof}
For $R>1$, differentiate the right-hand side of \cref{eq:beta-R}:
\begin{equation}\label{eq:threshold-derivative} \frac{d}{dR}\frac{R}{1+2\ln R}=\frac{2\ln R-1}{(1+2\ln R)^2}. \end{equation}
The unique minimum occurs at $R=\sqrt e$, where the value is $\sqrt e/2$.
\end{proof}

\subsection{Proof of the main theorem}

\begin{proof}[Proof of \cref{thm:main}]
The claim is trivial for $c\le1$, so write $c=1+\beta$ with $0<\beta<\sqrt e/2$ and suppose that a randomized algorithm is $c$-competitive.
By \cref{lem:marginalization}, its expected behavior is a deterministic monotone fractional trajectory with the same cost on every public prefix.
By \cref{cor:necessary-integrals}, this trajectory would force $I_h\le\beta$ for every $h$.
Lemma \ref{lem:monotone-limit} would then produce an integrable fixed point, while \autoref{lem:classification} and \autoref{prop:critical-beta} would require $\beta\ge\sqrt e/2$, a contradiction.
Hence some finite $h$ satisfies $I_h>\beta$.
Proposition \ref{prop:finite-witness} converts this strict violation into a finite oblivious bipartite input on which \cref{eq:main-ratio} holds.
Because bipartite graphs form a subclass of general graphs, the same lower bound applies in the general-graph model.
\end{proof}

\section{The Exact Barrier of the Recurrence}\label{sec:barrier}

The preceding argument proves every constant strictly below $1+\sqrt e/2$.
We now show that the recurrence itself cannot cross the critical value.

Set $R_*=\sqrt e$ and $\beta_*=\sqrt e/2$, and define
\begin{equation}\label{eq:critical-profile} D_*(t)=\begin{cases} 1-\dfrac{\sqrt e}{2}t, & 0<t\le e^{-1/2},\\[2mm] \dfrac{1}{2}\ln\dfrac{\sqrt e}{t}, & e^{-1/2}\le t\le\sqrt e,\\[2mm] 0, & t\ge\sqrt e. \end{cases} \end{equation}

\begin{proposition}[Critical fixed point]\label{prop:critical-fixed-point}
The function $D_*$ satisfies
\begin{equation}\label{eq:critical-fixed-equations} T_{\beta_*}D_*=D_*,\qquad D_0\le D_*,\qquad \int_0^{\infty}D_*(t)\,dt=\frac{\sqrt e}{2}-\frac{1}{4\sqrt e}<\beta_*. \end{equation}
Consequently, the iterates at $\beta=\beta_*$ satisfy $D_h\le D_*$ and $I_h<\beta_*$ for every finite $h$.
\end{proposition}

\begin{proof}
Direct integration of the two nonzero pieces in \cref{eq:critical-profile} verifies the fixed-point identity on $0<t\le R_*$.
For $t>R_*$, the untruncated expression is
\begin{equation}\label{eq:outside-critical-support} \frac{\beta_*}{2}\left(\frac{1}{t}-\frac{t}{R_*^2}\right)\le0, \end{equation}
so the positive part is zero there.
The inequality $D_0\le D_*$ is equality up to $1/R_*$; on the remaining interval where $D_0$ is positive, $D_*'(t)=-1/(2t)\ge-\beta_*=D_0'(t)$, and afterwards $D_0=0$.
The integral in \cref{eq:critical-fixed-equations} follows by elementary calculation.
Finally, order preservation of $T_{\beta_*}$ and induction give $D_h\le D_*$.
\end{proof}

Thus $1+\sqrt e/2$ is the exact threshold of the homogeneous complete-bipartite multi-alternation recurrence.
This is a statement about the proof framework, not an upper bound for online vertex cover.
A stronger problem lower bound would require a hard family that is not dominated by the critical profile in \cref{eq:critical-profile}.

\section{Discussion and Open Problems}\label{sec:discussion}

The lower bound isolates a source of online cost that is already present in bipartite graphs: irrevocable marginal weight accumulates across alternating complete-bipartite blocks before the adversary reveals which side will serve as the cheap offline cover.
The finite extension lemma shows that no infinite input or adaptive access to a random execution is needed.
The fixed-point analysis then converts arbitrarily deep alternation into the closed-form threshold $\sqrt e/2$.

The result narrows, but does not close, the gap to the approximately $1.901$ algorithm of \citet{WangWong2015}.
It also leaves open whether randomized integral algorithms on general graphs can exploit or suffer from constraints beyond the edge relaxation.
Odd-cycle inequalities provide additional restrictions on marginals, but inserting internal edges also changes the offline comparator, so those inequalities do not automatically strengthen the present recurrence.
A non-bipartite improvement may require a construction in which mistakes from several partially revealed minimum covers overlap and accumulate without an equal increase in the offline comparator.

At the critical value, the explicit profile $D_*$ leaves a positive slack between its integral and $\beta_*$.
This does not imply the existence of an algorithm with ratio $1+\sqrt e/2$; it only explains why further alternations of the same homogeneous complete-bipartite form stop improving the lower bound.
Closing the remaining gap will require either a nonhomogeneous recurrence, a different family of prefixes, or an algorithmic invariant that matches the stronger lower bound.

\appendix

\section{A Directed-Rounding Regression Certificate}\label{app:certificate}

The analytic proof above does not depend on numerical computation.
As a reproducibility check, a directed integer implementation lower-bounds the truncated recurrence on a uniform grid.
Let $n$ be the number of grid cells per unit interval, let $T$ be a finite cutoff, and suppose the step function $d_h$ is a certified lower bound for $D_h$ on every cell $[j/n,(j+1)/n]$.
Before the first budget violation, $d_h$ is nonincreasing and has integral at most $\beta$.
The raw update
\begin{equation}\label{eq:raw-update} G_h(t)=1-\beta t+t\int_{1/t}^{\infty}d_h(u)\,du \end{equation}
is then nonincreasing within every target cell.
Indeed, with $q=1/t$, its derivative there is $-\beta+H_{d_h}(q)+q d_h(q)\le0$, because $q d_h(q)\le\int_0^q d_h$ and $\int_0^\infty d_h\le\beta$.
Thus the cell's right endpoint gives a valid lower bound.

At the right endpoint $t=a/n$, where $a=j+1$, set $k=\lfloor n^2/a\rfloor$.
Retaining the fractional part of the reciprocal cell gives the exact nonnegative tail contribution
\begin{equation}\label{eq:endpoint-tail} t\int_{1/t}^{T}d_h(u)\,du=\frac{a\sum_{r=k+1}^{nT-1}d_{h,r}+d_{h,k}\bigl(a(k+1)-n^2\bigr)}{n^2}. \end{equation}
Storing all cell values on a common integer scale and rounding only this nonnegative contribution downward preserves rigor at every iteration.

With
\begin{equation}\label{eq:certificate-parameters} \beta=\frac{82435}{100000}=0.82435,\qquad n=24000,\qquad T=2, \end{equation}
the first certified violation occurs at $h=754$:
\begin{equation}\label{eq:certificate-values} L_{753}=0.809481094841418\ldots\le\beta,\qquad L_{754}=0.885466040873602\ldots>\beta. \end{equation}
This check is conservative because the complete tail beyond $T=2$ is discarded.
It supplies a near-threshold regression test for the recurrence and corresponds to at most $756$ alternating blocks; the exact theorem follows from the fixed-point argument, not from these parameters.

\bibliographystyle{plainnat}
\bibliography{references}

@article{BuchbinderNaor2009,
  author = {Niv Buchbinder and Joseph Naor},
  title = {Online Primal-Dual Algorithms for Covering and Packing},
  journal = {Mathematics of Operations Research},
  volume = {34},
  number = {2},
  pages = {270--286},
  year = {2009}
}

@article{DemangePaschos2005,
  author = {Marc Demange and Vangelis Th. Paschos},
  title = {On-Line Vertex-Covering},
  journal = {Theoretical Computer Science},
  volume = {332},
  number = {1--3},
  pages = {83--108},
  year = {2005}
}

@article{KarlinEtAl1994,
  author = {Anna R. Karlin and Mark S. Manasse and Lyle A. McGeoch and Susan Owicki},
  title = {Competitive Randomized Algorithms for Nonuniform Problems},
  journal = {Algorithmica},
  volume = {11},
  number = {6},
  pages = {542--571},
  year = {1994}
}

@inproceedings{GamlathEtAl2019,
  author = {Buddhima Gamlath and Michael Kapralov and Andreas Maggiori and Ola Svensson and David Wajc},
  title = {Online Matching with General Arrivals},
  booktitle = {Proceedings of the 60th Annual IEEE Symposium on Foundations of Computer Science},
  pages = {26--37},
  publisher = {IEEE Computer Society},
  year = {2019}
}

@inproceedings{KarpEtAl1990,
  author = {Richard M. Karp and Umesh V. Vazirani and Vijay V. Vazirani},
  title = {An Optimal Algorithm for On-Line Bipartite Matching},
  booktitle = {Proceedings of the 22nd Annual ACM Symposium on Theory of Computing},
  pages = {352--358},
  publisher = {ACM},
  year = {1990}
}

@inproceedings{TangZhang2024,
  author = {Zhihao Gavin Tang and Yuhao Zhang},
  title = {Improved Bounds for Fractional Online Matching Problems},
  booktitle = {Proceedings of the 25th ACM Conference on Economics and Computation},
  pages = {279--307},
  publisher = {ACM},
  year = {2024}
}

@misc{Tang2026Matching,
  author = {Tang, Zhihao Gavin},
  title = {Optimal Competitive Ratio of Two-Sided Online Bipartite Matching},
  year = {2026},
  eprint = {2602.18049},
  archivePrefix = {arXiv},
  primaryClass = {cs.DS}
}

@misc{TangZhang2026,
  author = {Zhihao Gavin Tang and Yuhao Zhang},
  title = {A Tight Bound on Online Vertex Cover under Edge Arrivals},
  year = {2026},
  eprint = {2608.04994},
  archivePrefix = {arXiv},
  primaryClass = {cs.DS}
}

@inproceedings{WangWong2015,
  author = {Wang, Yajun and Wong, Sam Chiu-wai},
  title = {Two-Sided Online Bipartite Matching and Vertex Cover: Beating the Greedy Algorithm},
  booktitle = {Automata, Languages, and Programming (ICALP 2015), Part I},
  series = {Lecture Notes in Computer Science},
  volume = {9134},
  pages = {1070--1081},
  publisher = {Springer},
  year = {2015}
}

\end{document}